\documentclass{article}
\usepackage{arxiv}
\usepackage[utf8]{inputenc} 
\usepackage[T1]{fontenc}    
\usepackage{hyperref}       
\usepackage{url}            
\usepackage{booktabs}       
\usepackage{amsfonts}       
\usepackage{nicefrac}       
\usepackage{microtype}      
\usepackage{graphicx}
\usepackage[backend=biber,style=alphabetic,sorting=ynt]{biblatex}
\usepackage{csquotes}
\usepackage{doi}
\usepackage{tabularx}
\usepackage{lipsum}
\usepackage{wrapfig}
\usepackage{listings}
\usepackage{mathrsfs}
\usepackage{amsmath}
\usepackage{amssymb}
\usepackage{bbm}
\usepackage{bm}
\usepackage{esvect}
\usepackage{upgreek}

\usepackage{amsthm}
\theoremstyle{definition}
\newtheorem{definition}{Definition}[section]
\newtheorem{example}{Example}[section]
\newtheorem{proposition}{Proposition}[section]
\newtheorem{remark}{Remark}[section]

\newtheorem{theorem}{Theorem}[section]
\newtheorem{lemma}{Lemma}[section]
\newtheorem{corollary}{Corollary}[section]

\usepackage{hyperref}

\title{On the equivalence of holding cost and response time for evaluating performance of queues.}

\author{Dylan Solms\\
	Department of Decision Sciences\\
	University of South Africa\\
	South Africa \\
	\texttt{62652257@mylife.unisa.ac.za} \\
}

\renewcommand{\headeright}{Discussion}
\renewcommand{\undertitle}{Discussion}
\renewcommand{\shorttitle}{Queuing Theory}

\hypersetup{
pdftitle={On the equivalence of holding cost and response time for evaluating performance of queues.},
pdfsubject={Queuing Theory},
pdfauthor={Dylan Solms},
pdfkeywords={Queuing Theory, Limits, Steady-state,Ergodic,Little's Law,Modelling},
}

\begin{document}
\maketitle

\begin{abstract}
This self-contained discussion relates the long-run average holding cost per unit time to the long-run average response time per customer in a $G/G/1$ queue with no assumption made on the order of service. The only restriction established is that the system be ergodic. This is achieved using standard queuing theory. The practical relevance of such a result is discussed in the context of simulation output analysis as well as through an application to formulating a Markov Decision Process that minimises long-run average response time per customer. 
\end{abstract}

\tableofcontents

\section{Introduction}

Made-made systems often face the dilemma of provisioning finite resources to persistent demand. These systems are usually stable in the sense that the long-term rate of supply exceeds the long-term rate of demand. Despite this, the latter might exceed the former for durations due to the innate stochastic behaviour found in both processes. The result is that queues start to form; a phenomenon all too familiar to every-day life. Furthermore, an exigent issue arises as unattended demand starts to incur a cost.

Demand is customarily produced by discrete entities, referred to as customers. Resources are provisioned to these by a server. Being able to quantify the cost of unattended demand is non-trivial as it allows for the queuing-system in question to undergo performance evaluation \cite{harchol2013performance}, optimal-design \cite{stidham_optimal_queue_design} and optimal-control \cite{cassandras_book,BertsekasVol2}. Subsequently, quality of service (QoS) can be assessed and optimised for before and during operation.

In prescribing controls to a queuing system it is natural to select actions that \emph{minimise} the average long-term amount of time each customer spends waiting. \emph{Delay} $\mathcal{D}$ is the amount of time a customer spends in the buffer of a queue before departing to be served while \emph{response time}\footnote{Other terms include turnaround time, time ins system and sojourn time.} $\mathcal{R}$ is the sum of delay and service duration experienced by an entity \cite{harchol2013performance} i.e. total amount of time spent in the system. Hence, an appealing and intuitive objective function (OF) to minimise is the average long-run response time \emph{per customer} $\bar{\mathcal{R}}_n$. 

The literature pertaining to optimal control of queuing systems does not use this. The most popular OF has been \emph{holding cost} \cite{BertsekasVol2,duenyas1996heuristic,hofri1987optimal}. This is also referred to as \emph{work-in-progress} (wip) \cite{van_eekelen}. More specifically, the average long-run holding cost \emph{per unit time} $\bar{\mathcal{H}}_t$ has manifested as the OF of choice. Such a performance metric is arguably easier to use than $\bar{\mathcal{R}}_n$ as only the time-stamp $\tau \in \mathbb{R}_{\geq 0}$ of a queue arrival or queue departure to service has to be recorded in conjunction with the the queue length $n(\tau) \in \mathbb{N}_{0}$. Holdings cost has enjoyed much success in applications to Markov Decision Processes where it has allowed to solve for an optimal policy over queue lengths (see chapter 9 of \cite{cassandras_book}, chapter 1.4 of \cite{BertsekasVol2}, the papers \cite{suk1991optimal,moustafa1996optimal,fernandesscheduling,haijema_traffic_MDP} as well as the PhD thesis \cite{bhulai2002markov}). Variants of holding cost have also been developed such as Harrison's reward function \cite{harrison1975priority} which can be interpreted as the holding cost saved from not holding a job any longer. Such a variant is nontrivial as a reward function is required if one wishes to use a \emph{Gittin's Index} \cite{gittins1979bandit} as a policy for controlling the system as a Multi-Armed Bandit \cite{MAB_gittins_tutorial}. Such an approach has been successful in \emph{scheduling} from which the $\mu c$ rule has been derived (see example 1.4.2 of \cite{BertsekasVol2}). Harrison's reward function is also used in scheduling \emph{polling systems with switch-over durations} \cite{hofri1987optimal,duenyas1996heuristic}.

This paper finds $\bar{\mathcal{H}}_t$ not to be as interpretable as $\bar{\mathcal{R}}_n$, Moreover, $\bar{\mathcal{R}}_n$ would appear to be a far more intuitively appealing OF to minimise as it pertains to the individual customer whereas $\bar{\mathcal{H}}_t$ concerns the system as a whole.

The main result of this paper is that it establishes the performance measures to be related by the queue's \emph{mean arrival rate} $\lambda$ such that $\bar{\mathcal{H}}_t = \lambda \bar{\mathcal{R}}_n$.

\section{Background}\label{section:Background}

\subsection{Sample paths and observations}\label{section: sample paths}

The performance metrics are computed using data obtained from a queuing system $\theta$ in the form of trajectories which are sequences of the form $\mathcal{T}(\omega,\theta) =  \left\{(\tau,n(\tau)) \right\}_{\tau=T_i}^{{T_f}}$ where $T_i$ and $T_f$ are the initial and final time-stamps, respectively. An infinite amount of possible trajectories exist where $\omega \in \Omega$ denotes a single realisation. The mean or expected trajectory is obtained by integrating over $\Omega$. The trajectory can be obtained from observing a real-world scenario or from a computer-simulation. In the latter, $\omega$ can be interpreted as a specific random seed for the sequence of pseudo-random numbers to be used. An objective function $F$ is dependent on $\mathcal{T}(\omega)$ as well as a policy $\pi: \mathbb{N}_{0} \to \mathcal{A}$ where $\mathcal{A}$ is some set of actions e.g. controlled service rate. This dependence can be denoted as $F\left( \mathcal{T}(\omega,\theta),\pi \right)$. Throughout this paper, $\mathcal{T}(\omega,\theta)$ and $\pi$ will be omitted from notation in pursuit of clarity and convenience with. It will only be recalled when necessary. Most of the time, expressing only $\omega$ will be sufficient enough to express the fact that trajectories are being produced by the same model $\theta$ under fixed policy $\pi$. Lastly, varying $\theta$ and $\pi$ will produce the response surface of $F$ under $\omega$. 

\begin{definition}[\textbf{equivalent under optimisation}]\label{def: equivalent under optimisation}
If $F_i$ and $F_j$ are two response surfaces with critical points\footnote{Global/local maxima and minima as well as inflection points and constant regions.} at the same argument locations then $F_i$ and $F_j$ will be regarded as \emph{equivalent under optimisation}. 
\end{definition}

A practical interpretation of definition~\ref{def: equivalent under optimisation} is that if a deterministic gradient-based optimisation algorithm were to start at initial arguments $(\theta_0,\pi_0)$ on $F_i$ and produce a $K$-length search sequence $\mathcal{S} = \{(\theta_0,\pi_0),\cdots,(\theta_k,\pi_k),\cdots,(\theta_K^*,\pi_K^*) \}$ in order to reach a critical point\footnote{No assumption is made that this is a very desirable point or eve less a global extreme.} $(\theta_K^*,\pi_K^*)$ then the same sequence will be iterated by starting the algorithm at $(\theta_0,\mu_0)$ on $F_j$. In fact, if the algorithm were to start at any $(\theta_k,\pi_k) \in \mathcal{S}$ on $F_j$ then it would reach $(\theta_K^*,\pi_K^*)$ via $\mathcal{S}_{k:K}$.
\\
\begin{example}[\textbf{Equivalent under optimisation}]
The following two transforms produce $F_j$ from $F_i$ such that they are equivalent under optimisation. The \emph{additive constant} transform
\begin{equation}
    F_j(\alpha) = F_i(\alpha) + C \label{eq: additive constant}
\end{equation}
where $C\in \mathbb{R}$ and the \emph{constant positive product} transform
\begin{equation}
    F_j(\alpha) = F_i(\alpha)\times \phi \label{eq: constant product}
\end{equation}
where $\phi \in \mathbb{R}_{>0}$.
\end{example}

\begin{example}[\textbf{Not equivalent under optimisation}]
The following two transforms produce $F_j$ from $F_i$ such that they are \emph{not} equivalent under optimisation. Adding constants to the arguments
\begin{equation}
    F_j(\alpha) = F_i(\alpha+ C)
\end{equation}
and multiplying/scaling the arguments by a constant
\begin{equation}
    F_j(\alpha) = F_i(\alpha\times \phi).
\end{equation}
\end{example}
The constant product transform (\ref{eq: constant product}) will play an important role in the results of section~\ref{section: relationship}.

\subsection{Time averages, ensemble averages and ergodicity}

Two types of averages exist: time and ensemble averages. If a system is \emph{ergodic} then the two are equivalent \cite{harchol2013performance}. This paper will be interested in using time-averages.\\

\begin{definition}[\textbf{Time-average \cite{harchol2013performance}}]
For a time-dependent quantity $f(t,\omega)$ and sample path $\omega \in \Omega$, the time-average follows: 
\begin{equation}
    \bar{f}_{time} =  \lim_{t\to \infty } \left\{ \frac{1}{t} \int_{0}^t f(\tau,\omega) \, d\tau \right\}
\end{equation}
\end{definition}
 \begin{definition}[\textbf{Ensemble-average \cite{harchol2013performance}}] For a time-dependent continuous quantity $f(t) \in \mathbb{R}$, the ensemble-average follows
 \begin{eqnarray}
     \bar{f}_{ens}  & = & \lim_{t \to \infty} \mathbb{E}\left[f(t) \right] \\
     & = & \int_{0}^\infty z \times \lim_{t \to \infty}\left\{ P\left(f(t)=z \right)\right\} \, dz
 \end{eqnarray}
 wheres for a time-dependent discrete quantity $g(t) \in \mathbb{Z}$
 \begin{equation}
     \bar{g}_{ens}  =  \sum_{z=0}^{\infty} z \times \lim_{t \to \infty}\left\{ P\left(g(t)=z \right)\right\}.
 \end{equation}
 \end{definition}

\begin{definition}[\textbf{Ergodic \cite{harchol2013performance}}]\label{def: ergodic}
An ergodic system is \emph{irreducible}, \emph{positive recurrent} and \emph{aperiodic}. Irreducibility refers to the fact that a system can reach any state $j \in \mathcal{X}$ from any other state $i \in \mathcal{X}$ such that the initial state chosen is not of long-term importance. Positive recurrence ensures that each state is visited infinitely often in the limit such that the process stochastically restarts itself (renewal). Each renewed process counts as an independent sample hence the importance of this property in establishing equivalence between ensemble and time-averages. Aperiodicity refers to the system being independent of time-steps. This ensures that the ensemble average exists.
\end{definition}

\subsection{Total response time and holding cost}

\begin{definition}[\textbf{Total holding cost}]
The \emph{total} holding cost up to time $t \in \mathbb{R}_{\geq 0}$ from the start of the observation is given as
\begin{equation}
    \mathcal{H}(t) = c \int_{T_i}^t n(\tau) \, d\tau \label{eq: total holding cost}
\end{equation}
where $c \in \mathbb{R}_{>0}$ is a cost weight and $n(\tau)\in \mathbb{N}_0$ the queue length. The queue length consists of the length of the buffer and the customer in service.
\end{definition}
 Figure~\ref{fig: delay vs wip} shows the progression of holding cost over time where (\ref{eq: total holding cost}) is the area. The colors highlight the fact that different customers contribute to this performance metric over the lifetime of the observation. The magenta and cyan ticks on the $x$-axis denote arrival and departure instances/events, respectively. As the start of observations do not always correspond to the beginning of the system's lifetime $T_i>0$, an initial portion of the true/actual trajectory may be unobserved. Finite length observations also lead to later portions of the actual trajectory remaining unobserved. Unobserved quantities are accentuated by a hatch texture. Holding cost remains uncomplicated by unobserved portions of the trajectory as it pertains to the system as a whole and not to the individual customer.

\begin{figure}[ht]
    \centering
    \includegraphics[width=0.6\textwidth]{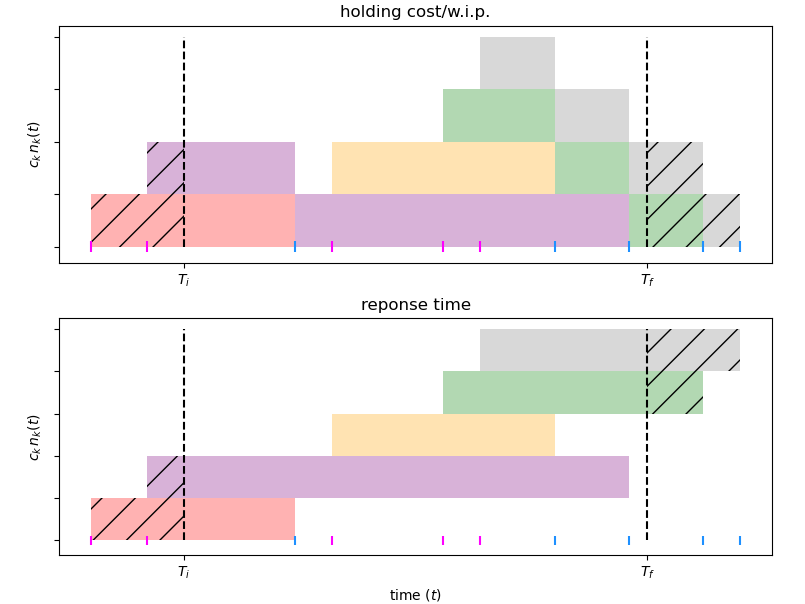}
    \caption{Holding cost and response time on the same sample path}
    \label{fig: delay vs wip}
\end{figure}

Response time is complicated by this distinction as it concerns the individual. Figure~\ref{fig: delay vs wip} shows that certain customers have a portion of their $\mathcal{R}$ unaccounted for by the observation. This motivates the need recognise two types of \emph{total} response time: observed $\mathcal{R}^{obs}(t)$ and actual $\mathcal{R}^{act}(t)$. Actual response time includes an unobserved amount $\mathcal{R}^{un}(t,\omega) = \mathcal{R}_i^{un}(\omega) + \mathcal{R}_f^{un}(t,\omega) $ that was truncated from customers already in the system by $T_i$ and those that remain in the system after $T_f$. As a result, the observed metric underestimates the actual amount $\mathcal{R}^{obs}(t) \leq \mathcal{R}^{act}$(t).

\begin{remark}
The unobserved response is sample dependent, hence the inclusion of $\omega$. In particular, $\mathcal{R}_i^{un}(\omega)$ is a constant determined by where $T_i$ is placed. If the ergodic system has run long enough before being observed such that it has entered a steady-state regime then $ \mathcal{R}_i^{un}$ will also be a steady state value. Moreover, it will be a time-average that is equivalent to an ensemble average due to the ergodicity assumption. In such a case, it will be independent of $\omega$.
\end{remark}

Let $\Delta_i = t_i^D - t_i^A \geq 0$ denote the response time experienced by the $i^{th}$ customer where $t_i^A$ is the arrival time into the buffer of the queue and $t_i^D$ the time of departure from the system. Furthermore, $A(t,\omega)$ and $D(t,\omega)$ are counting processes of \emph{observed} arrivals and departures such that $n(t) = A(t)- D(t) + n\left( T_i \right)$ whereas $\{A(t,\omega)\}$ and $\{D(t,\omega)\}$ denote sets that record \emph{all} customer arrivals and departures over the trajectory starting at $T_i$ and up to time $t>T_i$. Lastly, $\{A(T_i)\}$ contains \emph{unobserved} arrival information of the $n\left(T_i\right)$ existing customers. With this notation in place, total actual and observed versions can be defined.
\begin{definition}[\textbf{Total actual response time}]
\begin{equation}
    \mathcal{R}^{act}(t) = \sum_{i \in \{A(t)\}\cup \{A(T_i)\}}  \Delta_i \, c \label{eq: total actual delay}
\end{equation}
\end{definition}

\begin{definition}[\textbf{Total observed response time}]
\begin{eqnarray}
    \mathcal{R}^{obs}(t) & =& \sum_{j \in \{A(T_i)\cap\{D(t)\}\}}  (t_j^D - T_i)\, c + \sum_{j \in \{A(T_i)\setminus\{D(t)\}\}}  (t - T_i)\, c + \nonumber\\
    & & \sum_{j \in\{A(t)\}\cap\{D(t)\}} \Delta_j \, c + \sum_{j \in\{A(t)\}\setminus\{D(t)\}} (t - t_j^A) \, c \label{eq: total observed response} \\
    & = & \sum_{j \in \{A(t)\}\cup \{A(T_i)\}} \left( \min\left(t_j^D,t\right) - \max\left(t_j^A,T_i\right)  \right) \, c \label{eq: total observed delay}.
\end{eqnarray}
\end{definition}

Thus far neither $\mathcal{H}(t),\mathcal{R}^{act}(t)$ or $\mathcal{R}^{obs}(t)$ have made any assumptions regarding the service order of customers. This is illustrated in figure~\ref{fig: delay vs wip} where the yellow customers does not adhere to first-come-first-serve (FCFS) order. While FCFS is commonly used \cite{harchol2013performance,cassandras_book}, this paper will keep the discussion and results relevant to any service order in an \emph{open system} \cite{harchol2013performance}.

\subsection{Average long-run response time and holding cost}
Two types of long-run averages are of interest: \emph{per unit time} $\bar{X}_t$ and \emph{per customer} $\bar{X}_n$. These will also be referred to as \emph{time-averages} and \emph{count-averages}, respectively.
\\
\begin{definition}[\textbf{long-run average per unit time}]
\begin{eqnarray}
    \bar{X}_t & = & \lim_{T \to \infty}\left\{ \frac{1}{T} \int_{T_i}^{T_i+T} X(\tau) \, d\tau\right\} \label{eq: average per unit time}
\end{eqnarray}
\end{definition}
The total number of customers seen by the system is $N(t) = \left|\{A(t)\}\cup \{A(T_i)\}\right| = n(T_i) + A(t)$ for $t > T_i$.
\\
\begin{definition}[\textbf{long-run average per customer}]
\begin{eqnarray}
    \bar{X}_n & = &  \lim_{T \to \infty}\left\{ \frac{1}{N(T)} \int_{T_i}^{T_i+T} X(\tau) \, d\tau\right\} \label{eq: average per customer}
\end{eqnarray}
\end{definition}

Substituting $\mathcal{H}, \mathcal{R}_{obs}$ and $\mathcal{R}_{act}$ for $\int_{T_i}^{T_i+T} X(\tau) \, d\tau$ in (\ref{eq: average per unit time}) yields $\bar{\mathcal{H}}_t$, $\bar{\mathcal{R}}_t^{obs}$ and $\bar{\mathcal{R}}_t^{act}$ while (\ref{eq: average per customer}) grants $\bar{\mathcal{H}}_n, \bar{\mathcal{R}}_n^{obs}$ and $\bar{\mathcal{R}}_n^{act}$.

\subsection{Convergence of random variables}\label{section: random variable convergence}

If $\{X_n(\omega)\}_{n\in\mathbb{N}}$ denotes a sequence of random variables then convergence to $\mu_X$ will be established using \emph{almost-sure-convergence} and \emph{convergence-in-probability}.
\\
\begin{definition}[\textbf{Almost-sure-convergence \cite{harchol2013performance}}]\label{def: almost sure convergence}
A sequence of random variables $\{X_n(\omega)\}_{n\in\mathbb{N}}$ converges to $\mu_X$ if
\begin{equation}
    \forall \delta > 0: \quad P\left(\omega:\, \lim_{n \to \infty} \left\{\mid X_n(\omega) - \mu_X \mid\right\} > \delta  \right) = 0 \label{eq: almost sure convergence}
\end{equation}
which can be expressed as 
\begin{equation}
    X_n \xrightarrow{a.s.} \mu ,\quad \mbox{as } n \to \infty. \nonumber
\end{equation}
\end{definition}

\begin{definition}[\textbf{Convergence-in-probability \cite{harchol2013performance}}]\label{def: convergence in probability}
A sequence of random variables $\{X_n(\omega)\}_{n\in\mathbb{N}}$ converges to $\mu_X$ if
\begin{equation}
    \forall \delta > 0: \quad \lim_{n \to \infty} \left\{P\left(\omega:\,  \mid X_n(\omega) - \mu_X \mid > \delta  \right)\right\} = 0 
\end{equation}
which can be expressed as 
\begin{equation}
    X_n \xrightarrow{P} \mu ,\quad \mbox{as } n \to \infty. \nonumber
\end{equation}
\end{definition}
These two types of convergence are illustrated in figure~\ref{fig:convergence}.
\\
\begin{figure}[ht]
    \centering
    \includegraphics[width=0.5\textwidth]{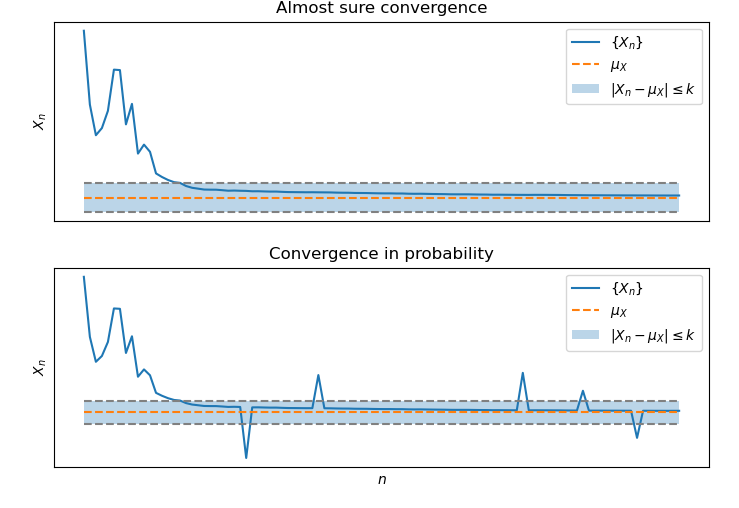}
    \caption{Random variable convergence}
    \label{fig:convergence}
\end{figure}
\\
Almost-sure-convergence is more strict in the sense that once it approaches $\mu_X$ then it continues to do so. The sequence is said to persist in behaving well. Probability in convergence allows for a limited amount of deviation or bad behaviour as long as the probability of observing such transgression is less than or equal to $\delta$. 

If instead $\{X_n(\omega)\}_{n\in\mathbb{N}}$ is a sequence of i.i.d. random variables that permit a sequence of interest $\{Y_n(\omega)\}_{n\in\mathbb{N}} = \left\{n^{-1}\sum_{i=1}^n X_i(\omega)\right\}_{n\in\mathbb{N}}$ then the above definitions can be used to guarantee convergence to $\mu_Y = \mathbb{E}[Y]$. Using almost-sure-convergence results in the \emph{Strong Law of Large Number} (SLLN) while convergence-in-probability translates to the \emph{Weak Law of Large Number} (WLLN) \cite{harchol2013performance}.

\subsection{Renewal-Reward Theory}\label{section: renewal-reward}

In definition~\ref{def: ergodic} it was established that ergodic processes restart themselves such that these renewed process can be considered as i.i.d. random variables. The points at which they restart are called \emph{renewal points}. A $G/G/1$ queue\footnote{In Kendall's notation, the first $G$ denotes generally distributed positive inter-arrival durations from a single source while the second refers to service durations. The one refers to a single server} has renewal points whenever an arrival enters an empty queue with idling server. A \emph{renewal cycle} is the interval between two such consecutive points $\mathscr{R}_i = \left[t_{i}^{\mathscr{R}},t_{i+1}^{\mathscr{R}}\right)$ where $t_{i}^{\mathscr{R}} \in \mathbb{R}_{\geq0}$. For the $G/G/1$ queue this would consist of a \emph{busy period} $B_i = \left[\underline{B}_i,\overline{B}_i\right)$ plus its neighbouring \emph{idle period} $I_i = \left[\underline{I}_i,\overline{I}_i\right)$ where $\overline{B}_i = \underline{I}_i$ such that $\mathscr{R}_i = \left[\underline{B}_i,\overline{I}_i\right)$. This is illustrated by figure~\ref{fig:renewals}.

\begin{figure}[ht]
    \centering
    \includegraphics[width=0.6\textwidth]{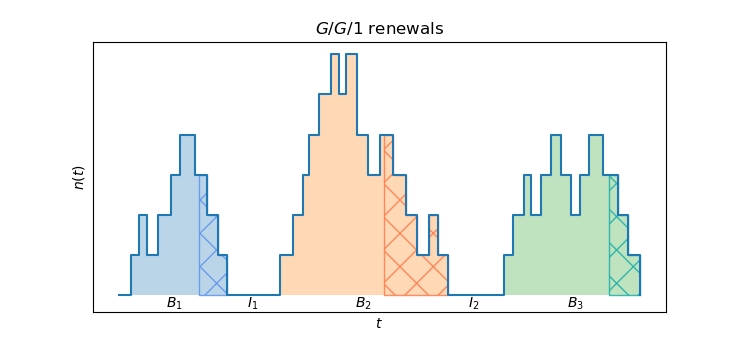}
    \caption{$G/G/1$ queue renewals consist of busy $B_i$ and idle $I_i$ periods}
    \label{fig:renewals}
\end{figure}

Renewal-reward theory is mainly used in obtaining time-averages from convenient ensemble averages.

\begin{definition}[\textbf{Renewal-reward time-average \cite{harchol2013performance}}]
If $F_R(t)$ denotes some \emph{cumulative}\footnote{In this context cumulative refers to successive additions of positive or negative values.} reward function $F:\mathbb{R}_{\geq 0}\to \mathbb{R}$ then its time-average can be determined as 
\begin{equation}
    \lim_{t\to\infty}\left\{\frac{F_R(t)}{t}\right\} = \frac{\mathbb{E}_{\mathscr{R}}\left[F_R\right]}{\mathbb{E}\left[\mathscr{R}\right]} \label{eq: time-average renewal}
\end{equation}
where $\mathbb{E}_{\mathscr{R}}\left[F_R\right]$ denotes the expected total reward earned over a renewal interval and $\mathbb{E}\left[\mathscr{R}\right]$ is the expected length of such an interval.
\end{definition}

Similarly, renewal-reward theory can also compute count-averages.

\begin{definition}[\textbf{Renewal-reward count-average}\cite{cassandras_book}]
If $N(t)$ denotes some cumulative count function $N:\mathbb{R}_{\geq 0} \to \mathbb{N}_0$ then 
\begin{equation}
    \lim_{t\to\infty}\left\{\frac{F_R(t)}{N(t)}\right\} = \frac{\mathbb{E}_{\mathscr{R}}\left[F_R\right]}{\mathbb{E}_{\mathscr{R}}\left[N\right]} \label{eq: count-average renewal}
\end{equation}
where $\mathbb{E}_{\mathscr{R}}\left[N\right]$ denotes the expected count per renewal interval.
\end{definition}

\subsection{Inspection Paradox}\label{section: inspection paradox}

The departure process $D(t)$ has been presented to be a function of time. This has been done for notational simplicity. It is in fact a function of the arrival process $A(t)$, delay/waiting time $\mathcal{D}_j(\sigma)$, service duration $t_j^{\mu}$ and service order $\sigma$. As a result, the departure time of the $j^{th}$ customer is also dependent on these factors
\begin{equation}
    t_j^D = t_j^A + \mathcal{D}_j(\sigma) + t_\mu.
\end{equation}
These are all random variables. As such, the response of the $j^{th}$ customer is also a random variable
\begin{eqnarray}
    \Delta_j & = & t_j^D - t_j^A\nonumber\\
    & = & \mathcal{D}_i(\sigma) + t_j^\mu
\end{eqnarray}
where $\sigma$ has been kept as general as possible. The service order might be to select a random customer from the queue. 
\\
\begin{example}
In the case of FCFS, $W_j(\sigma)$ evolves recursively according to the \emph{Lindley equation} \cite{cassandras_book}
\begin{eqnarray}
    \mathcal{D}_j\left(\mbox{FCFS}\right) & = & \max\left\{0,\mathcal{D}_{j-1}\left(\mbox{FCFS}\right) + t_j^{\mu} + t_j^A - t_{j-1}^A    \right\}
\end{eqnarray}
which allows subsequent values of interest such as $t_j^D$ and $\Delta_j$ to be computed as well.
\end{example}

At the termination time of the simulation $T_f$, a non-empty queue $n\geq 1$ will have a customer in service and $(n-1)$ in the buffer. The latter are all experiencing delay and are in the $\mathcal{D}_k(\sigma)$ process. Upon inspection, it is intuitive to expect $t_j^\mu(T_f)$ to be distributed according to the service time distribution $F_{\mu}$ (with probability density function $f_{\mu}$). Such false intuition leads to the \emph{inspection paradox}.
\begin{definition}[\textbf{Inspection paradox \cite{harchol2013performance}}]
The \emph{observed} service duration $t_j^\phi$ of a customer at a time where at least one service has already been completed $t_1^{\mu}$ is not the same as the service duration of the $j^{th}$ customer. More specifically, $t_j^\phi$ is \emph{stochastically larger} that $t_j^\mu$ such that
\begin{eqnarray}
    P\left(t_j^\phi > t\right) & \geq & P\left(t_j^\mu > t\right), \quad \forall t > 0\\
    \therefore  1 - F_\phi(t) & \geq & 1 - F_\mu(t), \quad \forall t > 0.
\end{eqnarray}
\end{definition}
It is understood that $t_j^\phi$ is \emph{positively biased} by some value $\beta$ because the likelihood of observing this interval is proportional to its length i.e. larger lengths leads to larger observation likelihoods. In contrast, $t_j^\mu$ has no such dependencies. Knowledge of this issue, allows one to properly establish the expected \emph{age} of the interrupted service $\mathbb{E}\left[\tau_j^\phi  \right] \neq \mathbb{E}\left[\tau_j^\mu  \right]$ as well as its expected \emph{residual lifetime} $\mathbb{E}\left[r_j^\phi  \right]$. These expected values can be derived using \emph{renewal-reward theory} (see chapter 23.4 of \cite{harchol2013performance})
\begin{eqnarray}
     \mathbb{E}\left[ \tau^\phi \right] & = & \frac{\mathbb{E}\left[(t^\mu)^2\right]}{2\mathbb{E}\left[t^\mu  \right]}\label{eq: age inspection paradox}\\
     & = & \mathbb{E}\left[ r^\phi \right] \label{eq: residual lifetime inspection paradox}
\end{eqnarray}
which further allows $\beta$ to be gauged
\begin{eqnarray}
    \beta & = & \mathbb{E}\left[ \tau^\phi \right] + \mathbb{E}\left[ r^\phi \right] - \mathbb{E}\left[ t^\mu \right]\\
    & = & \frac{\mathbb{E}\left[ (t^\mu)^2 \right] - \left(\mathbb{E}\left[ t^\mu \right]\right)^2}{\mathbb{E}\left[ t^\mu \right]}\nonumber \\
    & = & \frac{\mbox{Var}\left(t^\mu\right)}{\mathbb{E}\left[ t^\mu \right]}\\
    & = & C_V^2\left(t^\mu  \right)\mathbb{E}\left[ t^\mu \right]
\end{eqnarray}
where $C_V^2$ is the coefficient of variance. Variation is what increases bias in the service duration. Such a result is of value as it provides insight as to why queues form even though the arrival rates $\lambda = 1/\mathbb{E}[t^\lambda]$ are less than the service rates $\mu = 1/\mathbb{E}[t^\mu]$. Due to variance, service intervals occur in short and long variants with the long variant being greater than $\mathbb{E}[t^\mu]$. An arriving customer is much more likely to observe a long interval if the server is busy. More long intervals translates to greater waiting time and more queue formation. As the system is ergodic and the mean service rate is greater than the mean arrival rate then short intervals should empty the queue sometime in the limit/long-run.
\\
\newpage
\section{The relationship between long-run average holding cost per time and long-run average response time per customer}\label{section: relationship}

The goal of this section is to establish a relationship between $\bar{\mathcal{H}}_t$ and $\bar{\mathcal{R}}_n$ in order to assess whether they are equivalent under optimisation. This would be of great value as it would mean that using the more convenient $\bar{\mathcal{H}}_t$ for optimisation purposes would yield arguments $\alpha^* = \mbox{argmin}\{\bar{\mathcal{H}}_t \}$ that also minimise the more intuitive and arguably desirable $\bar{\mathcal{R}}_n $. This would translate to optimising for the overall system also optimising for the individual. Furthermore, this section also seeks to expose when distinguishing between $\mathcal{R}^{act}$ and $\mathcal{R}^{obs}$ is of importance.

\subsection{Total transient cost}

The transient cost refers to the value of the objective function accumulated up to some finite time $T_f$. At this time, an ergodic system has not necessarily entered steady-state behaviour \cite{harchol2013performance}.
\\
\begin{proposition}\label{prop: obs delay equal holding cost}
For $t > T_i \geq 0$ the observed holding cost and observed response times are equal
\begin{equation}
    \mathcal{R}^{obs}(t) = \mathcal{H}(t).
\end{equation}
\end{proposition}
\begin{proof}
This might be immediately obvious from figure~\ref{fig: delay vs wip} where the total areas of the two metrics clearly match. To show this to be mathematically true, the definition of total holding cost (\ref{eq: total holding cost}) will suffice.
\begin{eqnarray}
    \mathcal{H}(t) & = & c\int_{T_i}^{t} n(\tau)\, d\tau \nonumber\\
    & = & c \int_{T_i}^{t}n(T_i) +  A(\tau) - D(\tau) \, d\tau \nonumber \\
    & = &c\int_{T_i}^{t} \left\{A(T_i)\right\}\setminus\left\{ D(\tau)\right\} +  \left\{A(\tau)\right\}\setminus\left\{ D(\tau)\right\} \, d\tau \nonumber \\
    & = & \sum_{j \in \{A(T_i)\cap\{D(t)\}\}}\left(  \int_{T_i}^{t_j^D} c \, d\tau \right) + \sum_{j \in \{A(T_i)\setminus\{D(t)\}\}}  \left(\int_{T_i}^{t} c\, d\tau \right)+ \nonumber\\
    & & \sum_{j \in\{A(t)\}\cap\{D(t)\}} \left(\int_{t_j^A}^{t_j^D} c\, d\tau \right) + \sum_{j \in\{A(t)\}\setminus\{D(t)\}} \left(\int_{t_j^A}^{t} c\, d\tau \right) \nonumber\\
    & = & \sum_{j \in \{A(T_i)\cap\{D(t)\}\}}  (t_j^D - T_i)\, c + \sum_{j \in \{A(T_i)\setminus\{D(t)\}\}}  (t - T_i)\, c + \nonumber\\
    & & \sum_{j \in\{A(t)\}\cap\{D(t)\}} \Delta_j \, c + \sum_{j \in\{A(t)\}\setminus\{D(t)\}} (t - T_j^A) \, c  \nonumber \\
    & = & \sum_{i \in \{A(t)\}\cup \{A(T_i)\}} \left( \min\left(t_i^D,t\right) - \max\left(t_i^A,T_i\right)  \right) \, c  \label{eq: observed delay sequence}\\
    & = & \mathcal{R}^{obs}(t) \nonumber
\end{eqnarray}
\end{proof}

\begin{proposition}
For $t > T_i \geq 0$, the actual total response time might exceed the total holding cost
\begin{equation}
    \mathcal{R}^{act}(t) \geq \mathcal{H}(t).
\end{equation}
\end{proposition}
\begin{proof}
\begin{eqnarray}
    \mathcal{R}^{act}(t) & = & \sum_{j \in \{A(t)\}\cup \{A(T_i)\}}  \Delta_j \, c \nonumber \\
    & = & \sum_{j \in \{A(t)\}\cup \{A(T_i)\}} \left(t_j^D - t_j^A \right) \, c \nonumber \\
    & = & \sum_{j \in \{A(t)\}\cup \{A(T_j)\}} \left( \min\left(t_j^D,t\right) - \max\left(t_j^A,T_i\right)  \right) \, c + \nonumber\\
    & & \sum_{j \in \{A(t)\}\cup \{A(T_j)\}} \left(t_j^D-t\right)\mathbbm{1}_{\{t_j^D>t\}} + \left(T_i-t_j^A\right)\mathbbm{1}_{\{T_i > t_j^A\}} \label{eq: actual delay sequence} \\
    & = & \mathcal{R}^{obs}(t) + \mathcal{R}^{un}(t)\nonumber\\
    & = & \mathcal{H}(t) + \mathcal{R}^{un}(t)\nonumber \hspace{5.1cm} [\mbox{Proposition}  \ref{prop: obs delay equal holding cost}] \\
    & \geq & \mathcal{H}(t) \nonumber
\end{eqnarray}
\end{proof}

\subsection{Total long-run cost}

The notion of a long-run cost $F$ requires limits to be taken up to infinity. It is natural to expect that $F = \lim_{t\to\infty}\{F(t)\} = \infty$ which would not seem to be useful. In practice, a long-run total cost is measured over a large finite time horizon $T^*$; large enough for an ergodic system to be in it stationary regime. Such a cost is finite. This section seeks to assess whether different total costs \emph{converge} in the limit such that $F_i(T^*) \approx F_j(T^*)$.

In what follows, the system is assumed to be ergodic. At this point it is important to notice that $\mathcal{R}^{un}(t)$ consists of a constant component $\mathcal{R}_i^{un} = \sum_{j \in \{A(t)\}\cup \{A(T_j)\}}\left(T_i-t_j^A\right)\mathbbm{1}_{\{T_i > t_j^A\}}$ and a time-dependent component $\mathcal{R}_i^{un}(t) = \sum_{j \in \{A(t)\}\cup \{A(T_j)\}} \left(t_j^D-t\right)\mathbbm{1}_{\{t_j^D>t\}}$. The constant component is determined by $\omega$ and where the observation starts $T_i$. The time-dependent component will change value but might be best approximated by its steady-state time-average.

The quantities of interest are all random variables. Their stochastic nature stems from the fact that $n(t), t_j^A$ and $t_j^D$ are all random variables. However, these quantities would also appear to be continuous functions of time. As to avoid any confusion, the functions of interest will be expressed as the sums of a sequence of random variables: holding cost and observed response by equation (\ref{eq: observed delay sequence}) and actual response by equation (\ref{eq: actual delay sequence}). The appeal of this approach is that the machinery of section~\ref{section: random variable convergence} can be used.

As $\mathcal{H}(t)= \mathcal{R}^{obs}(t)$, proving that $\mathcal{H} = \lim_{\to\infty}\left\{ \mathcal{H}(t)\right\} = \lim_{\to\infty}\left\{ \mathcal{R}^{obs}(t)\right\} = \mathcal{R}^{obs}$ is trivial. This paper will only show that long-term observed response time $\mathcal{R}^{obs}$ \emph{does not} converge to the long-term actual response time $\mathcal{R}^{act}$. Before doing so, proposition~\ref{prop: time-average unobserved response} is required which pertains to the aforementioned time-dependent component of unseen response time $\mathcal{R}^{un}_{f}(t)$.

\begin{proposition}\label{prop: time-average unobserved response}
If an ergodic system has reached its stationary regime then $\mathcal{R}^{un}_{f}(t)$, observed at some randomly selected $t$, should be represented by its steady-state value which can be computed using renewal-reward theory of section~\ref{section: renewal-reward} 
\begin{eqnarray}
    \lim_{t \to \infty} \left\{ \mathcal{R}_f^{un}(t,\omega) \right\} & = & \mathbb{E}\left[\mathcal{R}_f^{un} \right]\label{eq:steady state unobserved response}.
\end{eqnarray}
\end{proposition}

Figure~\ref{fig:renewals} shows $\mathcal{R}_f^{un}(t)$ as indicated by the hatch pattern. The vertical line corresponds to an observation time $t$. An observation is sampled \emph{once} from each renewal cycle. In this figure, all observations have been sampled during the busy phase $B_i$, however, these might have occurred during the idle phase $I_i$ as well. The latter is straightforward in that $\mathcal{R}_f^{un}(t)=0$ if $t\in I_i$. 

Through sampling an observation from each $\mathscr{R}_i$ once, $N(t)$ denotes the amount of renewal cycles observed and $\mathbb{E}_{\mathscr{R}}\left[N\right]=1$. As such, equation~(\ref{eq:steady state unobserved response}) holds as a special case of (\ref{eq: count-average renewal}). More importantly each $\mathcal{R}_f^{un}(t)$ sampled over an $\mathscr{R}_i$ depends on quantities that have time-averages which can be determined from renewal-reward theory. 

For an observation to sample a non-zero $\mathcal{R}_f^{un}(t)$, it needs to observe $n(t)\geq 1$ such that the server is busy and $n(t)-1$ customers are in the buffer. Such an event occurs with steady-state probability $\rho$ that can be computed as
\begin{eqnarray}
\rho & =  & P(t\in B_i)\nonumber \\
& = &\lim_{N\to\infty} \frac{1}{N}\sum_{i=1}^{N}\frac{|B_i|}{|B_i|+|I_i|}\\
& = & \lim_{N\to\infty} \frac{1}{N}\sum_{i=1}^{N}\frac{\overline{B}_i-\underline{B}_i}{\overline{I}_i - \underline{B}_i}
\end{eqnarray}
or by renewal-reward theory
\begin{eqnarray}
 \rho & = & \frac{\mathbb{E}[B]}{\mathbb{E}[\mathscr{R}]}.
\end{eqnarray}

It is common to refer to $\rho$ as the server \emph{utilisation} as it reflects the portion of the time that the server is performing work\footnote{A $G/G/1$ queue with $\rho = 1$ is unstable and will grow without bounds.}. The customer in service will complete its residual service $r_k^\phi$ (see equation \ref{eq: residual lifetime inspection paradox}) derived from $t_k^\phi$ and not $t_k^\mu$ due to the inspection paradox of section~\ref{section: inspection paradox}. The other $n(t)-1$ customers will receive non-biased service durations $t_j^{\mu}\sim F_{\mu}$. The steady-state queue length is determined as a time-average
\begin{eqnarray}
 \bar{n} & = & \lim_{t\to\infty}\left\{ \frac{n(t)}{t}  \right\}\\
 & = & \frac{\mathbb{E}\left[\int_{\mathscr{R}} n(t)\,dt\right]}{\mathbb{E}[\mathscr{R}]}\\
 & = & \mathbb{E}[n(t)].
\end{eqnarray}

Note that $\bar{n}$ takes into account that the queue spends a non-trivial amount of time being empty. Equation~(\ref{eq:steady state unobserved response}) can now be determined
\begin{eqnarray}
 \mathbb{E}\left[ \mathcal{R}_f^{un}\right] & = & \mathbb{E}\left[\mbox{Unfinished service}\right] + \mathbb{E}\left[\mbox{Service to perform}\right]\\
 & = & \rho\times\mathbb{E}\left[r_k^\phi\right] + \mathbb{E}\left[ \sum_{j=1}^{n(t)-1} t_j^{\mu}\right] \\
 & = & \frac{\rho\times \mathbb{E}\left[ \left(t^\mu\right)^2 \right]}{2\mathbb{E}\left[  t^\mu\right]} + \left(\mathbb{E}\left[n(t)\right]-1\right)\mathbb{E}\left[t^\mu \right] \\
& = &  \rho\times \bar{t}^\mu \times \frac{\left(C_V^2(t^\mu) + 1\right)}{2}  + \left(\bar{n}-1\right) \bar{t}^\mu.
\end{eqnarray}

The essence here is that proposition~\ref{prop: time-average unobserved response} and its above expression dictates that the time-dependent part of the unobserved response can safely by replaced by a constant $\mathbb{E}\left[\mathcal{R}_f^{un}\right]$ such that the total unobserved response time is a constant in the long-run. The key result of this section is obtained.
\\
\begin{proposition}
The long-term actual delay \emph{does not} converge to the long-term holding cost.
\end{proposition}
\begin{proof}
Consider the two sequence $\left\{\mathcal{R}_n^{obs}(\omega)\right\}$ and $\left\{\mathcal{R}_n^{act}(\omega)\right\}$ from the same sample path $\omega$. The sequence progresses each time a new rectangle is added as in figure~\ref{fig: delay vs wip}. Thus, $\mathcal{R}_n^{obs}$ is given by equation~(\ref{eq: total observed response}) and $\mathcal{R}_n^{act}= \mathcal{R}_n^{obs} + \mathcal{R}_n^{un}$ where $\mathcal{R}_n^{un}$ is the hatched area. To compare these two random sequences, almost-sure-convergence (see definition~\ref{def: almost sure convergence}) will be used.
\begin{eqnarray}
 P\left(\lim_{n \to \infty}\left\{ \left| \mathcal{R}_n^{act} -\mathcal{R}_n^{obs}  \right| > \delta\right\} \right) & = & P\left(\lim_{n \to \infty}\left\{ \mathcal{R}_n^{un} \right\} >\delta  \right)\nonumber\\
 & = & P\left( \mathbb{E}\left[  \mathcal{R}^{un} \right] > \delta \right)\nonumber\\
 & = & 
 \begin{cases}
 1 ,&   0<\delta < \mathbb{E}\left[  \mathcal{R}^{un} \right]\\
 0 , & \delta \geq \mathbb{E}\left[  \mathcal{R}^{un} \right]
 \end{cases}
\end{eqnarray}
Hence, they do not converge and will always differ by a constant $\mathbb{E}\left[  \mathcal{R}^{un} \right]$.
\end{proof}

\begin{corollary}\label{corol: limiting response}
If $\mathcal{R}^{obs}(\omega) =   \lim_{t \to \infty} \left\{ \mathcal{R}^{obs}(t,\omega) \right\}$ and $\mathcal{R}^{act}(\omega) =   \lim_{t \to \infty} \left\{ \mathcal{R}^{act}(t,\omega) \right\}$ for $\omega \in \Omega$ then 
\begin{equation}
    \mathcal{R}^{act}(\omega) = \mathcal{R}^{obs}(\omega) + \mathbb{E}\left[  \mathcal{R}^{un} \right]
\end{equation}
which by the additive constant transform~(\ref{eq: additive constant}) means that they are equivalent under optimisation.
\end{corollary}

\begin{corollary}
If $\mathcal{H}(\omega) = \lim_{t\to\infty}\left\{\mathcal{H}(t,\omega) \right\}$ then by corollary~\ref{corol: long-run response} and proposition~\ref{prop: obs delay equal holding cost}
\begin{equation}
    \mathcal{R}^{act}(\omega) = \mathcal{H}(\omega) + \mathbb{E}\left[  \mathcal{R}^{un} \right].
\end{equation}
\end{corollary}

\begin{corollary}\label{corol: long-run response}
As a practical application/interpretation of corollary~\ref{corol: limiting response}, consider finite $T^*$ that is large enough for the ergodic system to be in a stationary regime. The following holds exactly and allows for equivalence under optimisation
\begin{equation}
    \mathcal{R}^{act}(T^*,\omega) = \mathcal{R}^{obs}(T^*,\omega) + \mathcal{R}_i^{un}(\omega) + \mathcal{R}_f^{un}(T^*\omega) 
\end{equation}
while the following approximation may be useful
\begin{eqnarray}
    \mathcal{R}^{act}(T^*,\omega)&  \approx &  \mathcal{R}^{obs}(T^*\omega) + \mathbb{E}\left[  \mathcal{R}^{un} \right]. \\
    &  \approx &  \mathcal{H}(T^*\omega) + \mathbb{E}\left[  \mathcal{R}^{un} \right] \label{eq: useful approximation}
\end{eqnarray}
\end{corollary}

\begin{remark}
The size of $\mathbb{E}\left[  \mathcal{R}^{un} \right]$ may become asymptotically negligible when compared to $\mathcal{R}^{act}(T^*,\omega)$ and $\mathcal{R}^{obs}(T^*\omega)$ such that the two are approximately equal. In practical finite length simulations, a small difference will be present. The analysis of this section has explained and determined the source of this bias.
\end{remark}

\subsection{Long-term count and time averages}

This section is interested interested in showing $\bar{\mathcal{H}}_t \propto \bar{\mathcal{R}}^{\psi}_t$ where $\psi = \{ obs,act \}$ as to denote a relationship between either type of response time. In fact, it will be established that the type of response time chosen makes no difference as it leads to the same long-term averages.

\begin{lemma}\label{lemma: reponse time averages}
Long-run time-averages for actual and observed response time are equal such that
\begin{equation}
    \bar{\mathcal{R}}_t^{act} = \bar{\mathcal{R}}_t^{obs}
\end{equation}
\end{lemma}
\begin{proof}
\begin{eqnarray}
    \bar{\mathcal{R}}_t^{act}(\omega) & = & \lim_{t\to\infty}\left\{ \frac{1}{t} \mathcal{R}^{act}(t,\omega)\right\} \nonumber\\
     & = & \lim_{t\to\infty}\left\{ \frac{ \mathcal{R}^{obs}(t,\omega)}{t} + \frac{\mathcal{R}^{un}(t,\omega)}{t}\right\}\nonumber \\
    & = & \bar{\mathcal{R}}_t^{obs}(\omega) + \lim_{t\to\infty}\left\{  \frac{\mathcal{R}_i^{un}(\omega)}{t}\right\} + \lim_{t\to\infty}\left\{  \frac{\mathcal{R}_f^{un}(t,\omega)}{t}\right\}\nonumber\\
    & = & \bar{\mathcal{R}}_t^{obs}(\omega) +   \frac{\mathbb{E}\left[\mathcal{R}_f^{un}\right]}{\lim_{t\to\infty}\left\{t\right\}}\label{step: ergodic assumption} \\
    & = & \bar{\mathcal{R}}_t^{obs}(\omega)\nonumber
\end{eqnarray}
Step~(\ref{step: ergodic assumption}) assumes the system to be ergodic Hence, time-averages equal ensemble averages $\bar{X}_t(\omega) = \mathbb{E}\left[X\right]$ such that dependence on the sample path can be dropped which results in the desired $\bar{\mathcal{R}}_t^{act} = \bar{\mathcal{R}}_t^{obs}$.
\end{proof}

\begin{lemma}\label{lemma: reponse count averages}
Long-run count-averages for actual and observed response time are equal such that
\begin{equation}
    \bar{\mathcal{R}}_n^{act} = \bar{\mathcal{R}}_n^{obs}
\end{equation}
\end{lemma}
\begin{proof}
Recall that the number of customers observed by the system is a cumulative quantity $N(t) = n(T_i) + A(t)$ such that $\lim_{t\to\infty}\{N(t)\} = \infty$.
\begin{eqnarray}
    \bar{\mathcal{R}}_n^{act}(\omega) & = & \lim_{t\to\infty}\left\{ \frac{1}{N(t)} \mathcal{R}^{act}(t,\omega)\right\} \nonumber\\
     & = & \lim_{t\to\infty}\left\{ \frac{ \mathcal{R}^{obs}(t,\omega)}{N(t)} + \frac{\mathcal{R}^{un}(t,\omega)}{N(t)}\right\}\nonumber \\
    & = & \bar{\mathcal{R}}_n^{obs}(\omega) + \lim_{t\to\infty}\left\{  \frac{\mathcal{R}_i^{un}(\omega)}{N(t)}\right\} + \lim_{t\to\infty}\left\{  \frac{\mathcal{R}_f^{un}(t,\omega)}{N(t)}\right\}\nonumber\\
    & = & \bar{\mathcal{R}}_n^{obs}(\omega) +   \frac{\mathbb{E}\left[\mathcal{R}_f^{un}\right]}{\lim_{t\to\infty}\left\{N(t)\right\}}\label{step: ergodic assumption repeat} \\
    & = & \bar{\mathcal{R}}_n^{obs}(\omega)\nonumber
\end{eqnarray}
Step~(\ref{step: ergodic assumption repeat}) assumes the system to be ergodic. Hence, using equation~(\ref{eq: count-average renewal}) from renewal-reward theorem 
\begin{eqnarray}
    \bar{\mathcal{R}}_n^{obs}(\omega) & = & \lim_{t\to\infty}\left\{ \frac{\mathcal{R}^{obs}(t)}{N(t)} \right\}\nonumber\\
    & = & \frac{\mathbb{E}_{\mathscr{R}}\left[\mathcal{R}^{obs}\right]}{\mathbb{E}_{\mathscr{R}}\left[N\right]}\label{eq: application of count renewal}
\end{eqnarray}
which allows for independence from the sample-path $\omega$.
\end{proof}

\begin{remark}
The application of renewal-reward theory using $\mathcal{R}^{act}$ as in (\ref{eq: application of count renewal}) would not work as only the first and last renewal intervals would contain unobserved response time. Renewal-reward theory requires all renewal intervals to be treated as if they were i.i.d. random variables.
\end{remark}

\begin{proposition}\label{prop: general average reponse time}
Long-run count-average and time-average response time does not depend on the type of response time used. 
\end{proposition}
\begin{proof}
This is clear from lemma~\ref{lemma: reponse time averages} and lemma~\ref{lemma: reponse count averages}.
\end{proof}

\begin{lemma}\label{lemma: count average holding cost equal reponse}
Long-run count averages for holding cost and observed response are equal such that
\begin{eqnarray}
    \bar{\mathcal{H}}_n = \bar{\mathcal{R}}_n^{obs}.
\end{eqnarray}
\end{lemma}
\begin{proof}
\begin{eqnarray}
    \bar{\mathcal{R}}_n^{obs}(\omega) & = & \lim_{t\to\infty}\left\{ \frac{1}{N(t)} \mathcal{R}^{obs}(t,\omega)\right\} \nonumber\\
    & = & \lim_{t\to\infty}\left\{ \frac{1}{N(t)} \mathcal{H}(t,\omega)\right\}\nonumber\\
    & = & \bar{\mathcal{H}}_n(\omega)\nonumber
\end{eqnarray}
Recall that $\bar{\mathcal{R}}_n^{obs}$ is independent of $\omega$ as shown using renewal-reward theory in (\ref{eq: application of count renewal}). By virtue of this result and the proven equality, $\bar{\mathcal{H}}_n$ is also independent of $\omega$.
\end{proof}

\begin{corollary}\label{corol: Hn equals Tn}
Putting together lemma~\ref{lemma: reponse count averages} and lemma~\ref{lemma: count average holding cost equal reponse} results in
\begin{eqnarray}
    \bar{\mathcal{H}}_n = \bar{\mathcal{R}}_n^{act}.
\end{eqnarray}
\end{corollary}

By proposition~\ref{prop: general average reponse time}, the distinction between actual and observed response time can be dropped when working with average performance functions. From here on, $\bar{\mathcal{R}}_n$ and $\bar{\mathcal{R}}_t$ will be used.

\begin{theorem}\label{theorem: main}
\begin{equation}
    \bar{\mathcal{H}}_t = \lambda \bar{\mathcal{R}}_n \label{eq:main}
\end{equation}
\end{theorem}
\begin{proof}
The long-run average holding cost per unit time takes the form of (\ref{eq: average per unit time}).
\begin{eqnarray}
    \bar{\mathcal{H}}_t(\omega) & = & \lim_{t\to\infty}\left\{\frac{\mathcal{H}(t,\omega)}{t}\right\}.\label{eq: holding cost per unit time}
\end{eqnarray}
The long-run average holding cost per customer starts the proof.
\begin{eqnarray}
    \bar{\mathcal{H}}_n(\omega) & = & \lim_{t\to\infty}\left\{\frac{1}{N(t,\omega)}\mathcal{H}(t,\omega)\right\}\\
    & = & \lim_{t\to\infty}\left\{\frac{t}{N(t,\omega)}\times \frac{\mathcal{H}(t,\omega)}{t}\right\}\nonumber\\
    & = & \lim_{t\to\infty}\left\{\frac{t}{N(t,\omega)}\right\}\times\lim_{t\to\infty}\left\{ \frac{\mathcal{H}(t,\omega)}{t}\right\}\nonumber\\
    & = & \frac{1}{\lim_{t\to\infty}\left\{\frac{N(t,\omega)}{t}\right\}}\times\bar{\mathcal{H}}_t(\omega)\quad\quad [\mbox{substitute } (\ref{eq: holding cost per unit time})]\label{eq: theorem proof part1}
\end{eqnarray}
The remaining limit can be dealt with if the system is ergodic
\begin{eqnarray}
    \lim_{t\to\infty}\left\{  \frac{N(t,\omega)}{t}\right\} & = & \lim_{t\to\infty}\left\{  \frac{A(t,\omega)}{t}\right\} + \lim_{t\to\infty}\left\{  \frac{n(T_i,\omega)}{t}\right\}\nonumber\\
    & = & \frac{1}{\mathbb{E}\left[t^{\lambda}\right]}\\
    & = & \lambda \label{eq: theorem proof part2} 
\end{eqnarray}
where $\lambda$ is the $\omega$-independent time-average arrival rate\footnote{$\lambda$ should not be mistaken with the rate parameter of the Exponential distribution (commonly found in queuing theory. This paper assumes $t^{\lambda}$ to be generally distributed.}. By the ergodicity assumption, it is also the ensemble-average arrival rate which asserts the $\omega$-independence. Substituting (\ref{eq: theorem proof part2}) into (\ref{eq: theorem proof part1}) and removing $\lambda$ as a denominator yields
\begin{eqnarray}
    \bar{\mathcal{H}}_t(\omega) & = &  \lambda \bar{\mathcal{H}}_n(\omega)\\
    & = & \lambda \bar{\mathcal{R}}_n \quad\quad [\mbox{corollary }\ref{corol: Hn equals Tn}]\label{eq:rhs}
\end{eqnarray}
where the right-hand side of~(\ref{eq:rhs}) is $\omega$-independent. This implies the same for $\bar{\mathcal{H}}_t$ which proves the main result of this paper~(\ref{eq:main}).
\end{proof}

\newpage

\section{Questions and answers}

\paragraph{Question:} can theorem~\ref{theorem: main} be extended to other performance functions of interest?

\paragraph{Answer:} the well-known \emph{Little's Law} \cite{harchol2013performance} states that for ergodic systems
\begin{eqnarray}\label{eq: littles law}
    \bar{n}_t & = & \lambda \bar{\mathcal{R}}_t \times \frac{1}{c}
\end{eqnarray}
where $\bar{n}_t$ is the long-run average queue-length per unit time and $1/c$ removes the cost-coefficient. This allows for the following extensions
\begin{eqnarray}
    \bar{n}_t & = & \lambda \bar{\mathcal{H}}_t \times \frac{1}{c}\\
    & = & \lambda^2 \bar{\mathcal{H}}_n \times \frac{1}{c}\\
    & = & \lambda^2 \bar{\mathcal{R}}_n \times \frac{1}{c}\label{eq: response per customer littles law}
\end{eqnarray}
which are all equivalent under optimisation by constant product transform~(\ref{eq: constant product}). Optimising for $\bar{n}_t$ is useful when working with queuing systems where space constraints become of importance. Furthermore, such relationships can be exploited as a variance-reduction technique in the analysis of simulation output. \emph{Indirect estimation} \cite{glynn_1988_queue_simulation} computes a statistic of interest using a relationship such as (\ref{eq: response per customer littles law}) or Little's Law~(\ref{eq: littles law}) where one of the statistics, such as $\lambda$, is known with certainty or has been estimated before hand.

\paragraph{Question:} why is equation~(\ref{eq: useful approximation}) a useful approximation?

\paragraph{Answer:} in practice, the simulation time\footnote{Often called a simulation budget\cite{glynn_1988_queue_simulation}.} used for evaluating infinite-horizon cost functions is finite $T^*$. This is common when performing Monte-Carlo Rollouts in Reinforcement Learning \cite{suttonRLbook}. Equation~(\ref{eq: useful approximation}) allows for the actual response time to be approximated by a more computationally friendly and easier-to-code substitute. The use of holding costs means that a single sweep of integration can be performed over the simulated trajectory whereby a constant $\mathbb{E}[\mathcal{R}^{un}]$ is added. The integration reduces to the summation of $n(t_i)\times (t_{i+1}-t_i)$ over successive arrival and service completions where $t_i$ denotes the time-stamp of an event occurrence. This is much easier than tracking each individual customer and summing up its response time. Furthermore, a simulator might not disclose per customer information or give access to the unobserved response time in which case this approximation is the only reasonable option.

\paragraph{Question:} why has the dependence on a sample-path been stressed throughout?

\paragraph{Answer:} in applications such as Monte-Carlo Rollout \cite{suttonRLbook}, a single run is ideally performed to obtain an expected value. An expected value is independent of $\omega$. Hence, the presence of $\omega$ has been used to show when a single run is not necessarily sufficient to be used as an expected value. Note that the dependence of $\omega$ has vanished throughout this paper when a system was shown to be ergodic.

\paragraph{Question:} are there less strict versions of equivalence under optimisation?

\paragraph{Answer:} substituting a difficult objective function by a simpler one for optimisation purposes is common in the literature. In optimising \emph{parameterised} policies of discrete event dynamic systems (DEDS), the original performance metric is replaced by the same performance metric for a \emph{Stochastic Hybrid System} (SHS) \cite{cassandras2010_SHS} that is an abstraction of the parent DEDS. These SHS metrics can be very easy to compute. This simplicity is demonstrated in \cite{howell_2005}. Furthermore, while computing the performance metric from the sample path, the gradient of the performance metric can also be obtained through a \emph{Perturbation Analysis} estimator. The SHS substitute generally does not match the original objective function and is not equivalent under optimisation by the definition presented in this paper. However, these substitutes exhibit the fortunate outcome that the two performance metrics have extreme points found at reasonably the same parameter locations. For examples of this see page 698 of \cite{cassandras_book} or chapter 2.4.1 of this PhD thesis \cite{howell2006_thesis}. Common SHS approximations for queuing systems are \emph{Stochastic Fluid Models} (SFM) as introduced in chapter 11.9.1 of \cite{cassandras_book}. For objective functions that are very expensive to evaluate, a cheaper version is built using \emph{response surfaces} or \emph{surrogate models} \cite{forrester_surrogate_modelling_book}. The surrogate models are fit to limited input-output data\footnote{This is referred to as supervised learning in the machine learning literature.} obtained from the original model via a well chosen sampling plan. These surrogates are not strictly equivalent under optimisation but may share roughly the same critical points. The use of \emph{Gaussian Processes} as surrogate models allow for the approximated objective function to near equivalence under optimisation around the best extreme in the sampled range through sequential infill search. This is often referred to as \emph{Bayesian Optimisation}.

\paragraph{Question:} only the expected values for $\tau^{\phi}$ and $r^{\phi}$ were given in (\ref{eq: age inspection paradox}). Can probability density functions (p.d.f.) be derived for these? 

\paragraph{Answer:} the expected values in (\ref{eq: age inspection paradox}) were derived using renewal-reward theory (see section 23.4 of \cite{harchol2013performance}). Expected values can also be obtained through their density functions
\begin{eqnarray}
    f_{r^\phi}(t) & = & \frac{1-F_{t^\mu}(t)}{\mathbb{E}\left[t^\mu\right]}\\
    & = & \frac{\bar{F}_{t^\mu}(t)}{\mathbb{E}\left[t^\mu\right]}\\
    & = & \mu \bar{F}_{t^\mu}(t)\\
    & = & f_{\tau^\phi}(t)
\end{eqnarray}
where $\mu$ is the expected long-term service rate. These are derived in the \hyperref[appendix]{appendix}.
\paragraph{Question:} can the analysis of this paper be extended?

\paragraph{Answer:} this analysis should be repeated for the long-term discounted case. Little's Law has an extension
\begin{equation}
    \bar{n}_t^B = \lambda \bar{\mathcal{D}}_t
\end{equation}
where $\bar{n}_t^B$ is the average length of the buffer per unit time and $\bar{\mathcal{D}}_t$ is the average delay per unit time. Such performance is not interested in the customer receiving service. It would be useful to repeat the analysis of this paper using delay instead of response time. Moreover, establishing a relationship between the results for delay and response would allow for a complete framework.

\section{Examples}

A simulation experiment is first considered as to support the findings of section~\ref{section: relationship}. The second examples shows how these relationships allow for a Markov Decision Process to be formulated as to optimise long-run average response time per customer.

\subsection{Optimal service rate of a G/G/1 queue}
A company has a dedicated server to processing large computations that would otherwise not be suitable for the standard desktop computer. Only one job can be processed at a time due to all cores being devoted to the computations (this would allow for multiprocessing of the single job). The rest of the unserved jobs are placed in a buffer. The company has not decided on how to prioritise service in this buffer. Hence, FCFS has not been committed to as of yet. Service order is thus random. Inter-arrival times between jobs is generally distributed. However, it has been determined that jobs arrive at a rate of $\lambda$. Jobs are of varying size which means that the time taken to complete jobs follows a general distribution. This distribution has its mean service rate $\mu(\pi)$ managed by the company's policy $\pi$.

The service rate can tuned within the range $[0.95,0.125]$. A faster service rate means that jobs spend less time waiting. This translates to higher productivity and subsequent better revenue. However, higher service rates wear out the server and incurs a cost. This penalty has a per-unit-time cost rate of of the form $k_0\exp\left(-k_1 \mu(\pi)\right)$ where $k_0>0$ and $k_1>0$. High service rates are seen to be harmful to the server. The company would like to select a policy that minimises various types of long-run operational costs that include this penalty.

The company runs simulations under different settings of $\mu(\pi)$ while it tracks $\mathcal{H}(t)$, $\mathcal{R}^{obs}$ and $\mathcal{R}^{act}$. Based on historical data, the mean service rate is chosen to be $\lambda = 0.1$. Long-run simulations produce the following response surfaces as found in figure~\ref{fig:response surface simulation}.

\begin{figure}[ht]
    \centering
    \includegraphics[width=0.7\textwidth]{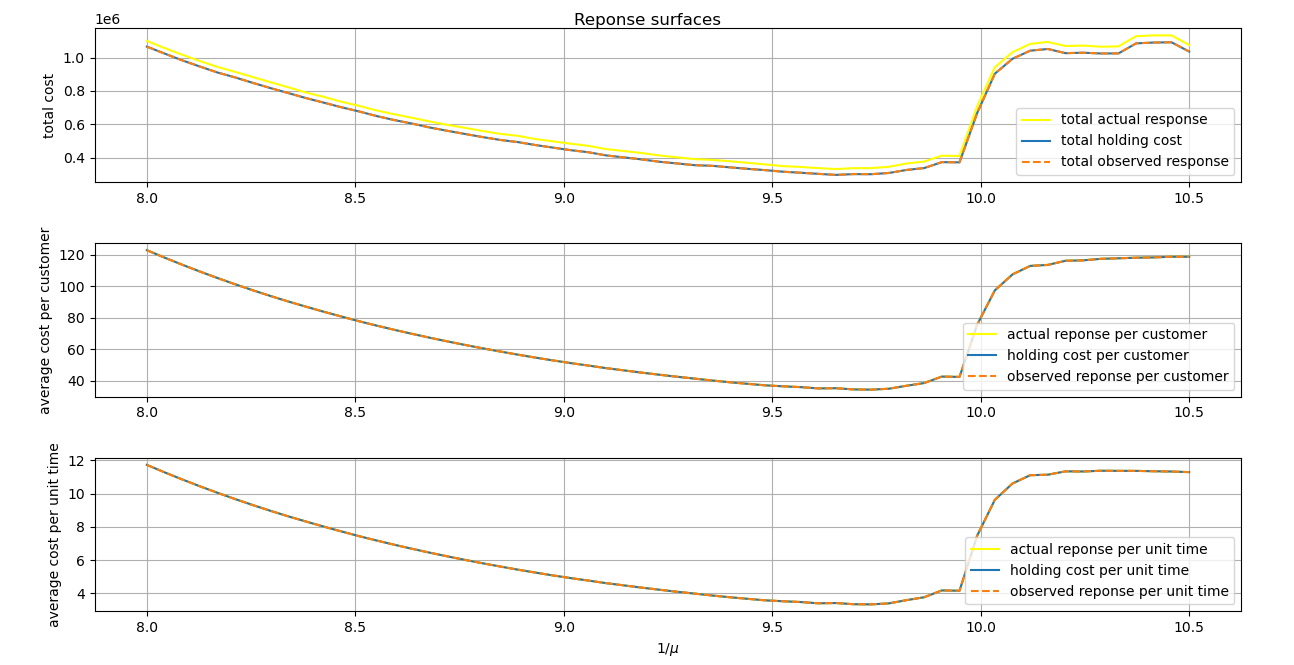}
    \caption{Response surfaces of $G/G/1$ queue with $\lambda=0.1$.}
    \label{fig:response surface simulation}
\end{figure}
The outcome is that all cost functions are equivalent under optimisation with $\mu\left(\pi^*\right) = 1/9.65$. 

\subsection{MDP service rate policy of a M/M/1 queue}\label{example: MDP}

The same company would like to obtain an optimal service rate by solving the problem as a Markov Decision Process (MDP) through the use of the \emph{Average-cost Bellman Equations} \cite{BertsekasVol2,puterman2014markov}
\begin{equation}
    J_\pi(x) = \bar{\rho}_\pi + c(x,\pi(x)) + \sum_{y\in \mathcal{X}} P(y\mid x,\pi(x))J_{\pi}(y)\label{eq: bellman}
\end{equation}
where $J_\pi(x)$ is a \emph{relative value function} over state $x$ using policy $\pi$ and forms part of the value function vector $\vec{J}_{\pi} = \left[J_{\pi}(x): x\in\mathcal{X}\right]$. The functions are relative in the sense that it is only the difference between them that matters in determining $\bar{\rho}_\pi$ and the optimal policy $\pi^*$. In practice, a state $\tilde{x}\in\mathcal{X}$ is chosen to be the \emph{distinguished state} such that $J_{\pi}(\tilde{x})=0$ which allows $n-1$ value functions and $\bar{\rho}_\pi$ to be solved for using $n$ Bellman equations~(\ref{eq: bellman}). Moreover, $c(x,\pi(x))$ is the cost incurred when in state $x$, $P(y\mid x,\pi(x))$ is a transition probability model and $\bar{\rho}_\pi$ is the average cost associated with $\pi$. Equation~(\ref{eq: bellman}) needs to be solved\footnote{Solving the Bellman equations for a fixed $\pi$ is called \emph{policy evaluation}.} for $J_\pi(x)$ and $\bar{\rho}_\pi$ and optimised\footnote{Greedy optimisation of $\pi$ is called \emph{policy improvement}.} for over $\pi$. This is achieved through Policy-Iteration, Value-Iteration or Linear Programming \cite{puterman2014markov}. This paper will not solve for the optimal policy. Instead focus will be on the issue of formulating an MDP, when solved for, would minimise $\bar{\mathcal{R}}_n$.

The company approximates the general service and arrival distributions with Exponential distributions that have rate parameters $\mu_{\pi(x)}$ and $\lambda$, respectively. The service rate parameters is to be chosen from a finite size interval grid $\mu_{\pi(x)} \in \left[\underline{\mu},\overline{\mu}\right]$

If $x\in \mathbb{N}_{0}$ is assumed to be the length of the queue, a transition model can be obtained through \emph{uniformisation} \cite{puterman2014markov,BertsekasVol2,cassandras_book}
\begin{equation}\label{eq:tpm}
    P(y\mid x,\pi(x)) = \begin{cases}
    \frac{\mu_\pi(x))}{\overline{\mu} + \lambda}, & y = x - 1\\
    \frac{\lambda}{\overline{\mu} + \lambda}, & y = x+1\\
    1 - \frac{\lambda+\mu_\pi(x))}{\overline{\mu} + \lambda}, & y = x
    \end{cases}
\end{equation}
along with per-stage-cost $c(x,\pi(x)) = cx/(\lambda+\overline{\mu})$. Such a MDP optimises $\bar{\mathcal{H}}_t$ and not $\bar{\mathcal{R}}_n$. A convenient transition model such as (\ref{eq:tpm}) that exhibits the Markov property is not possible for modelling response time. Response-time requires tracking each customer's time spent in the queue. This makes it history-dependent and not Markovian. However, theorem~\ref{theorem: main} shows $\bar{\mathcal{H}}_t$ and $\bar{\mathcal{R}}_n$ to be equivalent under optimisation such that solving for
\begin{equation}\label{eq: bellman holding cost}
    J_\pi(x) = \frac{\bar{\mathcal{H}}_t^\pi}{\lambda+\overline{\mu}} + \frac{cx}{\lambda+\overline{\mu}} + \sum_{y\in \mathcal{X}} P(y\mid x,\pi(x))J_\pi(y)
\end{equation}
would yield an optimal policy $\pi^*$ that also optimises $\bar{\mathcal{R}}_n$. After solving (\ref{eq: bellman holding cost}) for $\pi^*$ and $\bar{\mathcal{H}}_t^{\pi^*}$ then $\bar{\mathcal{R}}_t^{\pi^*} = \lambda \bar{\mathcal{H}}_t^{\pi^*}$.

\section{Conclusion}\label{section:conclusion}

The main objective of this paper was to show that $\bar{\mathcal{H}}_t$ and $\bar{\mathcal{R}}_n$ are equivalent under optimisation by a constant product relationship $\bar{\mathcal{H}}_t = \lambda \bar{\mathcal{R}}_n$ when the system of interest is ergodic. This objective was achieved without placing any restriction on the service order of customers. The probability distributions of the arrival and service durations were also kept as general as possible; only requiring that their domain be non-negative real numbers $\mathbb{R}_{\geq 0}$. The motivation behind this objective was provided by arguing $\bar{\mathcal{R}}_n$ as a more intuitive or desirable performance metric to optimise than the more practical $\bar{\mathcal{H}}_t$. This was illustrated in the MDP example of section~\ref{example: MDP}. Furthermore, such a relationship was discussed to have applications in variance reduction of simulation output. The work of this paper can be extended to establishing a similar relationship between the \emph{discounted} versions of $\mathcal{H}$ and $\mathcal{R}$ as well as a relationship between delay $\mathcal{D}$, holding cost of the buffer and buffer length. Lastly, while a $G/G/1$ queue was mostly discussed it should be emphasised that these results can be used in more complex models such as scheduling of polling systems or in routing to parallel queues.

\newpage

\section*{Appendix}\label{appendix}

The following derivation of observed residual life-times $r^\phi$ and observed process age $\tau^\phi$ is based on section 14.4 of~\cite{stewart2009probability}.
\\
Before proceeding with $r^\phi$ and $\tau^\phi$, the p.d.f. for $t^\phi$ is required. Its derivation starts by reasoning about $P\left(t\leq t^\phi \leq t + \Delta t\right)$ in that
\begin{eqnarray}
    P\left(t\leq t^\phi \leq t + \Delta t\right) & \propto & \mbox{length of the service interval}\\
    & \propto & t
\end{eqnarray}
as well as
\begin{eqnarray}
    P\left(t\leq t^\phi \leq t + \Delta t\right) & \propto & \mbox{frequency of service intervals with length }t\\
    & \propto & f_{t^\mu}(t)\,dt.
\end{eqnarray}
Furthermore, $\left(t\leq t^\phi \leq t + \Delta t\right) = f_{t^\phi}(t)\,dt$ which results in
\begin{eqnarray}
  f_{t^\phi}(t)\,dt & \propto & t \,  f_{t^\mu}(t)\,dt\\
  & = & \eta\, t\,  f_{t^\mu}(t)\,dt
\end{eqnarray}
where 
\begin{eqnarray}
  \eta & = & \left( \int_{0}^\infty t\, f_{t^\mu}(t)\,dt   \right)^{-1}\\
  & = & \frac{1}{\mathbb{E}\left[t^\mu\right]}\\
  & = & \mu
\end{eqnarray}
such that the desired p.d.f. is obtained
\begin{eqnarray}
    f_{r^\phi}(t) & = & \frac{t\times f_{t^\mu}(t)}{\mathbb{E}\left[t^\mu\right]}\\
    & = & \mu t \times f_{t^\mu}(t).
\end{eqnarray}
If a service duration $t^\phi$ has been observed/inspected then the service age $\tau^\phi$ at observation $T$ should be \emph{uniformly} distributed as $\tau^\phi \in [0,t^\phi]$. It can be reasoned that $P(\tau^\phi \leq t \mid t^\phi) = t/t^\phi = F_{\tau^\phi}(t\mid t^\phi)$. The derivative is taken with respect to~$t$.
\begin{eqnarray}
\frac{dF_{\tau^\phi}(t|t^\phi)}{dt} & =&  \frac{1}{t^\phi}\\
\therefore dF_{\tau^\phi}(t|t^\phi) &= &  \frac{dt}{t^\phi} \\
& =&  P(t \leq \tau^\phi \leq t + dt| t \leq t^\phi) \\
& =&  f_{\tau^\phi}(t|t^\phi) dt
\end{eqnarray}
Interest does not lie in the conditional. Hence $t^\phi$ needs to be marginalised out.
\begin{eqnarray}
P(t \leq \tau^\phi \leq t + dt) & =&  f_{\tau^\phi}(t) dt\\
& = & \int_{0}^{\infty} f_{\tau^\phi}(t|t^\phi)\,dt \, f_{t^\phi}(\phi) \mathbbm{1}_{\{t\leq \phi\}} \,d\phi\\
& =&  \int_{t}^{\infty} \frac{dt}{t^\phi} f_{t^\phi}(\phi) d\phi\\
& =&  \int_{t}^{\infty} \frac{dt}{\phi} \frac{\phi f_{t^\mu}(\phi)}{\mathbb{E}[t^\mu]} d\phi \\
& =&  \frac{1}{\mathbb{E}[\Delta]}\left( 1 - \int_{0}^{t} f_{t^\mu}(\phi)\,d\phi\,dt  \right) \\
& =&  \frac{1 - F_{t^\mu}(t)}{\mathbb{E}[t^\mu]}dt\\
& =&  \frac{\bar{F}_{t^\mu}(t)}{\mathbb{E}[t^\mu]}dt\\
\therefore f_{\tau^\phi}(t)&=&  \frac{\bar{F}_{t^\mu}(t)}{\mathbb{E}[t^\mu]}\\
& = & \mu \bar{F}_{t^\mu}(t)
\end{eqnarray}
By noticing that $r^\phi\sim\mbox{Uniform}(0,t^\phi)$ then the same derivation follows. Repeating this would be trivial. Instead, $\forall t \in \mathbb{R}_{\geq0}: f_{\tau^\phi}(t) = f_{r^\phi}(t)$ confirms that $\mathbb{E}\left[\tau^\phi\right] = \mathbb{E}\left[r^\phi\right]$ which agrees with the renewal-reward result from section~\ref{section: inspection paradox}.

\newpage
\medskip
\printbibliography

\end{document}